\documentclass[aip,reprint]{revtex4-1}

\usepackage{amsmath}
\usepackage{amssymb}
\usepackage[dvips]{graphicx}
\usepackage[OT4]{fontenc}
\usepackage{url}
\usepackage{bm}
\newcommand{\abs}[1]{\left| #1 \right|}

\newcommand{\okra}[1]{\left( #1 \right)}
\newcommand{\kwad}[1]{\left[ #1 \right]}
\newcommand{\klam}[1]{\left\{ #1 \right\}}

\newcommand{\dzi}[1]{\left\langle #1 \right\rangle}
\newcommand{\voc}[1]{\mathbf{V}[#1]}

\DeclareMathOperator{\const}{const}

\DeclareMathOperator{\sred}{\mathbf{E}}
\DeclareMathOperator{\card}{card}

\newcommand{\subseq}{\sqsubseteq}

\newtheorem{definition}{Definition}
\newtheorem{proposition}{Proposition}
\newtheorem{theorem}{Theorem}
\newtheorem{lemma}{Lemma}
\newenvironment{proof}{\begin{trivlist}\item[]
\noindent\textbf{Proof:}}{ $\Box$\par\end{trivlist}}

\draft 

\begin{document}


\title[Excess entropy in natural language]{Excess entropy in natural
  language: present state and perspectives}



\author{{\L}ukasz D\k{e}bowski}
\email[]{ldebowsk@ipipan.waw.pl}
\homepage[]{www.ipipan.waw.pl/~ldebowsk}
\thanks{}
\affiliation{Institute of Computer Science, Polish Academy of
  Sciences,
  Warszawa, 
  Poland}


\date{\today}

\begin{abstract}
  We review recent progress in understanding the meaning of mutual
  information in natural language. Let us define words in a~text as
  strings that occur sufficiently often.  In a~few previous papers, we
  have shown that a~power-law distribution for so defined words
  (a.k.a.\ Herdan's law) is obeyed if there is a~similar power-law
  growth of (algorithmic) mutual information between adjacent portions
  of texts of increasing length. Moreover, the power-law growth of
  information holds if texts describe a~complicated infinite
  (algorithmically) random object in a~highly repetitive way,
  according to an analogous power-law distribution. The described
  object may be immutable (like a~mathematical or physical constant)
  or may evolve slowly in time (like cultural heritage). Here we
  reflect on the respective mathematical results in a~less technical
  way. We also discuss feasibility of deciding to what extent these
  results apply to the actual human communication.
\end{abstract}



\maketitle 


\begin{quotation}
  In 1990, German engineer Wolfgang Hilberg published an
  article\cite{Hilberg90} where the graph of conditional entropy of
  printed English from Claude Shannon's famous work\cite{Shannon51}
  was replotted in log-log scale. Seeing a~dozen data points lie on
  a~straightish line, he conjectured that entropy of a~block of $n$
  characters drawn from a~text in natural language is roughly
  proportional to $\sqrt{n}$ for $n$ tending to infinity. Although
  this conjecture was not sufficiently supported by experiment or
  a~rational model, it attracted interest of a~few physicists seeking
  to understand complex systems.\cite{EbelingNicolis91,
    EbelingPoschel94,
    BialekNemenmanTishby01b, CrutchfieldFeldman03} As a~graduate in
  physics and a~junior computational linguist, I~found their
  publications in 2000. They stimulated me to ponder upon the
  interplay of randomness, order, and complexity in language. I~felt
  that better understanding of Hilberg's conjecture can lead to better
  understanding of Zipf's law for the distribution of
  words.\cite{Zipf65,Mandelbrot54} Using Hilberg's conjecture,
  I~wished to demonstrate clearly that the monkey-typing model,
  introduced to explain Zipf's law,\cite{Miller57} cannot account for
  some important purposes of human communication. However, it took
  a~few years to translate these intuitions into a~mature mathematical
  model.\cite{Debowski09,Debowski10,Debowski11b,Debowski11c} The model
  is presented here in an accessible way. I~also identify a~few
  problems for future research.
\end{quotation}

\section{Introduction}
The phenomenon of human language communication can be looked upon from
various perspectives. Respectively, these different points of view
give rise to different mathematical models, which are applied to human
language, studied for themselves, or used for different purposes.  The
most clear dichotomy of mathematical views onto language comes from
whether we look at individual sentences or above.

On the one hand, we may ask how human beings understand individual
sentences and what rules are obeyed in their composition. This
interest leads to elaborate theories of phonology, word morphology,
syntax, automata, formal and programming languages, mathematical logic
and formal semantics.\cite{Chomsky57,Montague70,HopcroftUllman79}
Although fragmented, these fields influence one another. Their common
feature is using discrete rather than numerical models. Thus, they may
be called non-quantitative linguistics (non-QL).

On the other hand, we may ask how sentences are chained into texts,
discourses, or collections of texts typically produced by humans. At
this level, rigid structures are less prominent and quantitative
analysis of data, done under auspices of quantitative linguistics (QL)
or corpus linguistics, forms the primary tool of description. However,
in spite of a~few remarkable observations like Zipf's\cite{Zipf65} or
Menzerath's\cite{Menzerath28} laws, QL has not established a~coherent
mathematical framework so far.\cite{KohlerAltmannPiotrovskij05}

Although communication between QL and non-QL is weak because of using
very different mathematical notions, quantitative reflection upon
language and difficulties of probabilistic modeling thereof inspired
a~few great mathematicians: A.\ Markov formulating the notion of
a~Markov chain,\cite{Markov13en2} C.\ Shannon establishing information
theory,\cite{Shannon48,Shannon51} B.\ Mandelbrot studying
fractals,\cite{Mandelbrot53,Mandelbrot54} and A.\ Kolmogorov
introducing algorithmic complexity.\cite{Kolmogorov65en2}

In this paper, I~present some conceptual framework for QL which
borrows heavily from information theory and yields a~macroscopic view
onto human communication. Because of the exposed connections among
mutual information, power laws, and emergence of hierarchical patterns
in data, I~suppose that my results may be interesting for researchers
in the domain of complex systems, who consider the power-law growth of
mutual information a~hallmark of complex
behavior.\cite{BialekNemenmanTishby01, BialekNemenmanTishby01b,
  CrutchfieldFeldman03} How to combine the `macroscopic' QL and the
`microscopic' non-QL into a~larger theory of language is a~different
problem. I~consider it worth pursuing but harder.

The central point of my paper is linking Herdan's law, an empirical
power law for the number of different
words,\cite{KuraszkiewiczLukaszewicz51en,Guiraud54,Herdan64,Heaps78}
with an intuitive idea that texts describe various facts in a~highly
repetitive and mostly logically consistent way.  Thus I~will discuss
a~proposition that can be expressed informally as follows:
\begin{itemize}
\item[(H)] If a~text of length $n$ describes $n^\beta$ independent
  facts in a repetitive way, where $\beta\in (0,1)$, then the text
  contains at least $n^\beta/\log n$ different words.
\end{itemize}
Proposition (H) has been formalized and proved by myself in a~series
of mathematical definitions and
theorems.\cite{Debowski10,Debowski11b,Debowski11c} It holds under an
appropriate quantification over $n$, which is a~combination of an
upper and a~lower limit over $n$.

Let me note that Proposition (H) can be also linked to the relaxed
Hilberg hypothesis. This conjecture says that (algorithmic) mutual
information between adjacent blocks of text of length $n$ is roughly
proportional to $n^\beta$.\cite{Hilberg90,EbelingNicolis91,
  EbelingPoschel94, BialekNemenmanTishby01,BialekNemenmanTishby01b,
  CrutchfieldFeldman03} Besides Proposition (H), I have formalized and
proved the following two propositions:
\begin{itemize}
\item[(H')] If a~text of length $n$ describes $n^\beta$ independent
  facts in a repetitive way, where $\beta\in (0,1)$, then mutual
  information between adjacent blocks of length $n$ exceeds $n^\beta$.
\end{itemize}
and:
\begin{itemize}
\item[(H'')] If mutual information between adjacent blocks of $n$ of
  length $n$ exceeds $n^\beta$, where $\beta\in (0,1)$, then the text
  of length $n$ contains at least $n^\beta/\log n$ different words.
\end{itemize}

The quantifications over $n$ in the formalizations of Propositions
(H') and (H'') are analogical as in the Proposition (H). For this
reason Proposition (H) does not follow from the conjunction of
Propositions (H') and (H''). All these propositions are, however,
true. The significance of the propositions is as follows. On the one
hand, Proposition (H') demonstates that Hilberg's hypothesis can be
motivated rationally. On the other hand, Proposition (H'') shows that
the hypothesis implies certain empirical regularities, such as
Herdan's law, even if there are problems with verifying Hilberg's
conjecture directly.


Consecutively, I~will introduce the concepts that appear in
Propositions (H), (H'), and (H'') and their formal statements. I~will
also discuss some related problems.  The composition of the paper is
as follows: In Section \ref{secBackground}, I~introduce the motivating
linguistic concepts. In Section \ref{secSynthesis}, I~discuss the
mathematical results. In Section \ref{secTheory}, I~reflect upon
limitations of these results as a~theory of human language or other
complex communication systems. Section \ref{secExperiments} contains
important remarks for researchers wishing to verify Hilberg's
hypothesis experimentally.  Section \ref{secConclusion} concludes the
paper.

\section{Ideas in the background}
\label{secBackground}

Before we embark on discussing formal models, I~should introduce some
linguistic playground on which the models will be built. First, I~will
recall empirical laws for the distribution of words. Second, I~will
introduce grammar-based codes as a~method of detecting word
boundaries. Third, Hilberg's hypothesis and its generalizations will
be presented. In the end, I~will discuss the idea of texts that
describe infinitely many facts in a~highly repeatable and logically
consistent way.

\subsection{Zipf's and Herdan's laws}
\label{ssecZipf}

A~few famous empirical laws of quantitative linguistics concern the
distribution of words. Amongst them, the Zipf-(Mandelbrot) law is the
most celebrated.\cite{Zipf65,Mandelbrot54} According to this law, the
word frequency $f(w)$ in a~text is an inverse power of the word rank
$r(w)$, i.e.,
\begin{align}
  \label{ZipfMandelbrot}
  f(w)\propto \kwad{\frac{1}{B+r(w)}}^{1/\beta}
  .
\end{align}
The frequency $f(w)$ of word $w$ is defined as the number of its
occurrences in the text, whereas the word rank $r(w)$ is the position
of $w$ on the list of words sorted by decreasing frequencies. Constant
$B$ is positive whereas constant $\beta\in (0,1)$ is close to $1$ for
$r(w)\lesssim 10^3\div 10^4$. For larger ranks this relationship is
breaks down and $\beta$ can drop much closer to $0$, depending on the
text composition.\cite{FerrerSole01b,MontemurroZanette02}

Zipf's law attracts attention of many theoreticians wishing to explain
it.  The most famous explanation of Zipf's law is given by the
`monkey-typing' model. In this explanation, the text is assumed to be
a~sequence of independent identically distributed (IID) variables
taking values of both letters and spaces and, as a~result, the
Zipf-Mandelbrot law is satisfied for strings of letters delimited by
spaces \cite{Mandelbrot54,Miller57}. 
Other known explanations involve, e.g., multiplicative processes
\cite{Simon55,Perline05} games,\cite{HarremoesTopsoe01} and
information theoretic
arguments.\cite{Manin08,FerrerSole03,FerrerDiazGuilera07,ProkopenkoOthers10}

In this paper, we will focus on a~certain corollary of the
Zipf-Mandelbrot law, namely a~relationship between the length of the
text and the number of different words therein. This relationship is
usually called Herdan's or Heaps' law in the English
literature.\cite{KuraszkiewiczLukaszewicz51en,Guiraud54,Herdan64,Heaps78}
It takes form of an approximate empirical power law
\begin{align}
  \label{HerdanHeaps}
  V\propto n^\beta
  ,
\end{align}
where $V$ is the number of different words and $n$ is the text length
(in characters). We can see that (\ref{HerdanHeaps}), up to
a~multiplicative logarithmic term, appears in the conclusion of
Propositions (H) and (H'').

The Herdan-Heaps law can be inferred from the Zipf-Mandelbrot law
assuming certain regularity of text growth.\cite{Khmaladze88,Kornai02}
In particular, if law (\ref{ZipfMandelbrot}) were satisfied exactly
then (\ref{HerdanHeaps}) would hold automatically. We have the
following proposition:
\begin{proposition}
  Let $N$ be the number of all words in the text and $V$ be the number
  of different words. If (\ref{ZipfMandelbrot}) is satisfied with
  $B=0$, $\beta$ constant, and $f(w)/N$ constant w.r.t. $N$ for the
  most frequent word $w$ then we have $V\propto
  N^\beta$.\cite{Kornai02}
\end{proposition}
\begin{proof}
  For the least frequent word $u$ we have the frequency $f(u)=1\propto
  V^{1/\beta}$. Hence the proportionality constant equals
  $V^{1/\beta}$. Thus for the most frequent word $w$ we have
  $f(w)=V^{1/\beta}$. Because $f(w)/N$ was assumed constant, we obtain
  $V\propto N^\beta$.
\end{proof}
In reality it happens that relationships (\ref{ZipfMandelbrot}) and
(\ref{HerdanHeaps}) are quite inexact and the best fit for
(\ref{HerdanHeaps}) yields a~$\beta$ smaller than for
(\ref{ZipfMandelbrot}).

In this article, I~propose another explanation of Herdan's law, which
is probabilistic. In any such explanation, two postulates are adopted
more or less explicitly. The first postulate concerns how words are
delimited in the text. The second postulate concerns what kind of
stochastic process is suitable for modeling the text. My explanation
of Herdan's law targets two modeling challenges:
\begin{enumerate}
\item Words can be delimited in the text even when the spaces are
  absent.
\item Texts refer to many facts unknown a~priori to the reader but
  they usually do this in a~consistent and repetitive way.
\end{enumerate}
The necessary notions will be explained consecutively. First, we will
revisit the concept of a~word. Second, we will address the properties
of texts.

\subsection{Detecting word boundaries with grammar-based codes}
\label{ssecAdmissible}


In this section, we will discuss how to delimit words in a~text and,
consecutively, how to count their number.  If we agree that texts are
sequences of characters taking values of both letters and delimiters
(such as spaces), the most obvious choice, suggested by orthographies
of many languages, is to define words as strings of letters separated
by delimiters. There are, however, kinds of texts or languages where
words are not separated by delimiters on a~regular basis (ancient
Greek, modern Chinese, or spoken English as a~speech signal).

Seeking for an absolute criterion for word boundaries, linguists
observed that strings of characters that are repeated within the text
significantly many times often correspond to whole words or set
phrases (multi-word expressions) like \emph{United
  States}.\cite{Wolff80,NevillManning96} Another important insight is
that the number of so detected `words' or `phrases' is a~thousand
times larger in texts produced by humans than in texts generated by
IID sources.\cite{Debowski07d}

A~particularly convenient way to detect words or sufficiently often
repeated strings is to use a~grammar-based code that minimizes the
length of a~certain text encoding.\cite{DeMarcken96,KitWilks99}
Grammar-based codes compress strings by transforming them first into
special grammars, called admissible grammars,\cite{KiefferYang00} and
then encoding the grammars back into strings according to a~fixed
simple method.  An admissible grammar is a~context-free grammar that
generates a~singleton language $\klam{w}$ for some string
$w$.\cite{KiefferYang00}

In an admissible grammar there is exactly one rule per nonterminal
symbol and the nonterminals can be ordered so that the symbols are
rewritten onto strings of strictly succeeding
symbols.\cite{KiefferYang00,CharikarOthers05} A~particular example of
an admissible grammar is as follows,
\begin{align*}
  \klam{
    \begin{array}{rl}
A_{1} \rightarrow
&\textbf{How\_much\_}A_{5}\textbf{\_w}A_{4}A_{2}A_{3}\textbf{,} \\
&\textbf{if}A_{2}\textbf{c}A_{4}\textbf{\_}A_{3}\textbf{\_}A_{5}\textbf{?} \\
A_{2} \rightarrow &\textbf{\_a\_}A_{5}A_{3}\textbf{\_} \\
A_{3} \rightarrow &\textbf{chuck} \\
A_{4} \rightarrow &\textbf{ould} \\
A_{5} \rightarrow &\textbf{wood} \\
    \end{array}
  },
\end{align*}
where $A_1$ is the initial symbol and other $A_i$ are secondary
nonterminal symbols.  If we start the derivation with symbol $A_1$ and
follow the rewriting rules, we obtain the text of a~verse:
\begin{verse}
  \it
How much wood \\
would a woodchuck chuck,\\
if a woodchuck \\
could chuck wood? 
%
\end{verse}

Although it cannot be seen in the short text above, secondary
nonterminals $A_i$ often correspond to words or set phrases in
compressions of longer texts. This correspondence is particularly good
if it is additionally required that nonterminals are defined as
strings of only terminal symbols.\cite{KitWilks99} For this reason,
the number of different words in an arbitrary text will be modeled in
the formalization of Propositions (H) and (H'') by the number of
different nonterminals in a~certain admissible grammar.

\subsection{Excess entropy and Hilberg's hypothesis}

Once I~have partly described how to detect and count ``words'' in an
arbitrary text, let us refine our ideas about texts typically produced
by humans.  There are justified opinions that such texts result from
a~very complicated amalgam of deterministic computation and
randomness\cite{Kilgarriff05} and this amalgam can be realized very
differently in particular texts, as mocked by D. Knuth.\cite{Knuth84}
To make these intuitions more precise, let us investigate entropy and
algorithmic complexity of texts.

Let us begin with entropy.  For a~probability space
$(\Omega,\mathfrak{J},P)$, the entropy of a~discrete random variable
$X$ is defined as
\begin{align}
  H_P(X):=-\sred_P{\log P(X)}
  ,
\end{align}
where $\sred_P$ is the expectation with respect to $P$ and random
variable $P(X)$ takes value $P(X=x)$ for $X=x$.

Subsequently, for a~discrete stationary process
$(X_i)_{i\in\mathbb{Z}}$, we define $n$-symbol entropy
\begin{align}
  \label{Hn}
  H_\mu(n):=H_P(X_{1}^{n})
  ,
\end{align}
where $X_{n}^{m}=(X_i)_{n\le i\le m}$ are blocks of variables and
$\mu=P((X_i)_{i\in\mathbb{Z}}\in \cdot)$ denotes the distribution
of $(X_i)_{i\in\mathbb{Z}}$ (i.e., $\mu(A)=P((X_i)_{i\in\mathbb{Z}}\in
A)$).  On the one hand, if the process is purely random, i.e., $X_i$
are IID variables, then $H_\mu(n)\propto n$. On the other hand, we
have $H_\mu(n)=\const$ if the process is in a~sense deterministic,
i.e., $X_i=f(X_{1}^{i-1})$.

Intuitively, texts written by humans are neither deterministic nor
purely random. This corresponds to a~particular behavior of entropy
$H_\mu(n)$. Some insight into this behavior can be obtained by asking
people to guess the next character of a~text given the context of $n$
previous characters. In one of his very first papers on information
theory,\cite{Shannon51} Shannon performed this experiment. As it was
later observed by Hilberg,\cite{Hilberg90} Shannon's data points obey
approximate relationship
\begin{align}
  \label{HilbergOrig}
  H_\mu(n)\propto n^\beta
\end{align}
for $\beta\approx \frac{1}{2}$, $n\lesssim 100$, and $H_\mu(n)$ being
an estimate of entropy of $n$ consecutive characters rather than the
entropy itself. Hilberg supposed that (\ref{HilbergOrig}) also holds
for much larger $n$, even for $n$ tending to infinity.

Some parallel research in entropy of texts in natural language
suggests that estimates of entropy depend heavily on a~particular
text\cite{HoffmanPiotrovskij79,Petrova73} and Shannon's guessing
method does not give precise estimates of entropy for large
$n$.\cite{CoverKing78} Thus Hilberg's conjecture (\ref{HilbergOrig})
should be modified and other ways of its justification should be
sought.

First of all, let us recall the concept of entropy rate
\begin{align}
  h_\mu:=\lim_{n\rightarrow\infty} \frac{H_\mu(n)}{n}
  .
\end{align}
For a~stationary process, conjecture (\ref{HilbergOrig}) implies
entropy rate $h_\mu=0$, which is equivalent to asymptotic determinism,
i.e., $X_1=f((X_i)_{i<1})$ almost surely.  Such asymptotic determinism
seems an unrealistic assumption. (But we may be wrong.)

Thus let us introduce block mutual information
\begin{align}
  E_\mu(n)
  &:= I_P(X_{1}^{n};X_{n+1}^{2n}) 
  \nonumber
  \\
  \label{En}
  &:= H_P(X_{1}^{n})+H_P(X_{n+1}^{2n})-H_P(X_{1}^{2n}),
\end{align}
called $n$-symbol excess entropy.\cite{CrutchfieldFeldman03}
$E_\mu(n)$ is a~convenient measure of complexity of discrete-valued
processes.  It vanishes for purely random processes and is bounded for
asymptotically deterministic ones.  Now let us observe that for
a~stationary process $(X_i)_{i\in\mathbb{Z}}$, we have
$H_P(X_{n+1}^{2n})=H_P(X_{1}^{n})$ and we obtain
\begin{align}
  \label{HilbergGen}
  E_\mu(n)\propto n^\beta
\end{align}
if (\ref{HilbergOrig}) is satisfied. We will call (\ref{HilbergGen})
the relaxed Hilberg conjecture. Notice that, unlike the case of
(\ref{HilbergOrig}), $h_\mu=0$ does not follow from
(\ref{HilbergGen}).

Thus, if proportionality (\ref{HilbergGen}) were actually satisfied
for any $n$ then texts in natural language could not be produced by
generalized `monkey-typing'. In the generalized `monkey-typing' model,
the text is generated by a~finite-state source a.k.a.\ a~hidden Markov
model. Indeed, if the finite-state source has $k$ hidden states then
$E_\mu(n)\le \log k$.\cite{ShaliziCrutchfield01}

Nonetheless, relationship (\ref{HilbergGen}) does not exhaust the
problem of reasonable generalizations.  Bluntly speaking, it seems
impossible to point out a~correct reference measure $P$ for texts in
natural language.\cite{Kolmogorov65en2} Although researchers in
linguistics happen to speak of entropies of a~single
text,\cite{HoffmanPiotrovskij79,Petrova73} this is an abuse of
concepts because entropy is a~function of a~distribution rather than
of a~text! To render the relaxed Hilberg conjecture for an individual
text $x_{1}^{n}$, we should use prefix algorithmic complexity
$H(x_{1}^{n})$ instead of entropy $H_P(X_{1}^{n})$. Formally, prefix
complexity $H(x_{1}^{n})$ is defined as the length of the shortest
self-delimiting program to generate text $x_{1}^{n}$.\cite{Chaitin75}

Thus for algorithmic mutual information
\begin{align}
  I(x_{1}^{n};x_{n+1}^{2n}) 
  \label{EnC}
  &:= H(x_{1}^{n})+H(x_{n+1}^{2n})-H(x_{1}^{2n}),
\end{align}
we will call relationship
\begin{align}
  \label{HilbergGenII}
  I(x_{1}^{n};x_{n+1}^{2n})\propto n^\beta
\end{align}
the relaxed Hilberg conjecture for individual texts.  This
relationship makes quite a~sense because in the probabilistic setting
we have
\begin{align}
  \label{BayesianEBounds}
  H_P(X_1^n)\le \sred_P H(X_1^n)\le H_P(X_1^n) + C^{P}_n
\end{align}
for any computable measure $P$ and $C^{P}_n=c_{P} + 2\log n$ with
$c_P<\infty$.\cite{LiVitanyi97} We remind that measure $P$ is called
computable when $P(X_1^n)$ can be computed given $X_1^n$ by a~fixed
Turing machine.  Under this assumption, law (\ref{HilbergGen}) follows
up to a~logarithmic correction if proportionality (\ref{HilbergGenII})
holds almost surely for a~fixed proportionality constant.

\subsection{Highly repetitive descriptions of a~random world}
\label{ssecDescriptions}

In this subsection, I~want to discuss the question \emph{why} texts
typically produced by humans diverge from both simple randomness and
determinism. This will provide a~justification for Hilberg's
conjecture. I~may point out three plausible reasons:
\begin{enumerate}
\item[A.] Texts attempt to describe an infinite collection of
  independent facts that concern either an immutable objective reality
  or an evolving historical heritage.
\item[B.] For some reasons, the immutable objective reality and the
  historical heritage are described in a~highly repetitive and mostly
  logically consistent way.
\item[C.] Any fact about the immutable objective reality can be
  inferred correctly given sufficiently many texts, according to
  a~fixed inference method, regardless of where we start reading.
\end{enumerate}
As I~will show in Subsection \ref{ssecDefTheorems}, the conjunction of
propositions A.--C. implies Hilberg's conjecture by a~formalization of
Proposition (H'). Thus let us inspect these statements closer.

As for postulate A., there exists a~collection of facts about an
immutable objective reality which is infinite and algorithmically
random. A~particular collection of that kind is given by the binary
expansion of halting probability $\Omega$. The expansion of $\Omega$
is an algorithmically random sequence and represents a~large body of
mathematical knowledge in its most condensed
form.\cite{Chaitin75,Gardner79} (Sequence $(x_i)_{i\in\mathbb{N}}$ is
called algorithmically random for algorithmic complexity
$H(x_1^n)\gtrsim n$.) Other plausible choices of immutable and
algorithmically random sequences are binary expansions of compressed
physical constants.

In contrast, the evolving historical heritage, which is primarily
described in texts, admits a~larger interpretation. Namely, this
heritage encompasses both the culture and the present state of the
physical world. We can also agree that the present state of the
physical world contains all material aspects of the culture.

To make these simple statements less abstract, let us mention a~few
examples of what falls under the evolving historical heritage. The
scope of culture covers: vocabulary and grammars of particular
languages, fictitious worlds described in novels, all heritage of
arts, humanities, science, and engineering. The present state of
physical world covers also all facts of biology, geography, and
astronomy, including those yet unknown.

To support postulate B., let us consider why the facts mentioned in
texts are described in a~highly repetitive and mostly logically
consistent way. This has more to do with the human nature than with
properties of the described world itself. As a~plausible reason,
I~suppose that human society develops communication structures to
maintain a~larger body of knowledge than any individual could manage
on his or her own.

Thus the the primary cause of repetition is probably the requirement
that knowledge is passed from generation to generation.  Moreover,
I~suppose that any human mind needs constant restimulation to remember
and reorganize the possessed knowledge. The result is that either in
fiction or in scientific writings, people prefer logically consistent
and directed narrations. This consistency also implies repetition.

To argue in favor of postulate C., let us observe the following.  In
the course of time, the historical heritage undergoes distributed
creation, accumulation, description, and lossy transmission from text
creators to text addressees. This should be contrasted with the
immutable objective reality, which can be discovered and described
independently by successive generations of text creators.

Thus it does not sound weird that every fact about the immutable
objective reality is described in some text ultimately and repeated
infinitely many times afterwards. Moreover, there should exist a~fixed
method of interpreting texts in natural language to infer these
facts. Such faculty is called human language competence in the
linguistic jargon and it allows knowledge to be passed from generation
to generation.

\section{Mathematical synthesis}
\label{secSynthesis}

The ideas presented in the previous section will now be synthesized as
an assortment of theorems and toy examples of stochastic
processes. This can be called a~formalization of Propositions (H),
(H') and (H''), mentioned in the Introduction. Namely, in a~series of
theorems I~will link Hilberg's conjecture with Herdan's law for
vocabulary size of admissible grammars and a~power law for the number
of facts that can be inferred from a~given text.  Afterwards, I~will
demonstrate a~few simple processes that exhibit all three laws. For
simplicity of argumentation, I~will discuss probabilistic Hilberg
hypothesis (\ref{HilbergGen}) rather than algorithmic one
(\ref{HilbergGenII}).
Respectively, both texts and facts will be modeled by random
variables.

In the following, symbol $\mathbb{N}$ denotes the set of positive
integers. For a~countable alphabet $\mathbb{X}$, the set of nonempty
strings is $\mathbb{X}^+:=\bigcup_{n\in\mathbb{N}} \mathbb{X}^n$ and
the set of all strings is
$\mathbb{X}^*:=\mathbb{X}^+\cup\klam{\lambda}$, where $\lambda$ stands
for the empty string. The length of a~string $w\in\mathbb{X}^*$ is
written as $\abs{w}$.

\subsection{Definitions and theorems}
\label{ssecDefTheorems}

In this subsection I~will show how Proposition (H) can be formalized.
First, the model of texts and facts is made precise. Second, the
model of words is elaborated. Third, I~present three previously
proved theorems\cite{Debowski11b} that link Hilberg's conjecture and
these two models.

Let $(X_i)_{i\in\mathbb{Z}}$ be a~discrete stochastic process with
variables $X_i:\Omega\rightarrow\mathbb{X}$, where $\Omega$ denotes
the event space. Process $(X_i)_{i\in\mathbb{Z}}$ models an infinite
text, where $X_i$ are characters if $\mathbb{X}$ is finite or
sentences if $\mathbb{X}$ is infinite. Moreover, let
$Z_k:\Omega\rightarrow\klam{0,1}$, where $k\in\mathbb{N}$, be
equidistributed IID binary variables. Variables $Z_k$ model facts
described in text. Their values (1=true and 0=false) can be
interpreted as logical values of certain systematically enumerated
independent propositions.

More specifically, let us assume that each fact $Z_k$ can be inferred
from a~half-infinite text according to a~fixed method if we start
reading it from an arbitrary position, like in postulate C. from
Subsection \ref{ssecDescriptions}. The method to infer these facts
will be formalized as certain functions $s_k$ which given a~text
predict whether the $k$-th fact is true or false. This leads to the
following definition.
\begin{definition}  
  \label{defiUDP}
  A~stochastic process $(X_i)_{i\in\mathbb{Z}}$ is called
  strongly nonergodic if there exists an IID binary process
  $(Z_k)_{k\in\mathbb{N}}$ with marginal distribution
  \begin{align}
    \label{Zk}
    P(Z_k=0)=P(Z_k=1)=\frac{1}{2}
  \end{align}
  and functions $s_k:\mathbb{X}^*\rightarrow \klam{0,1}$, where
  $k\in\mathbb{N}$, such that
  \begin{align}
    \label{UDPcondi}
    \lim_{n\rightarrow\infty} P(s_k(X_{t+1}^{t+n})=Z_k)&=1
    ,
    &
    &\forall t\in\mathbb{Z},\, \forall k\in\mathbb{N} 
    .
  \end{align}
\end{definition}

In the definition above, facts $Z_k$ are fixed for a~given realization
of $(X_i)_{i\in\mathbb{Z}}$ but they can be very different for
different realizations. I suppose that such probabilistic modeling of
both texts and facts reflects some properties of language, where
reality described in texts is most often created at random during text
generation and recalled afterwards. Under this assumption I~will
derive an average-case result.

Strong nonergodicity is indeed a~stronger condition than
nonergodicity. A~stationary process is strongly nonergodic when there
exists a~continuous random variable $\Theta:\Omega\rightarrow(0,1)$
measurable with respect to the shift-invariant
algebra.\cite{Debowski09} Such a~variable is an example of a~parameter
in terms of Bayesian statistics. Taking $\Theta=\sum_{k=1}^\infty
2^{-k} Z_k$ corresponds to a~uniform prior distribution on $\Theta$.

The number of facts described in text $X_{1}^{n}$ will be identified
with the number of $Z_k$'s that may be predicted with probability
greater than $\delta$ given $X_{1}^{n}$. That is, this number is
understood as the cardinality $\card{U_\delta(n)}$ of set
\begin{align}
\label{Un}
  U_\delta(n):= \klam{k\in\mathbb{N}:
    P\okra{s_k\okra{X_{1}^{n}}=Z_k}\ge \delta}
  .
\end{align}

There is also another condition for process $(X_i)_{i\in\mathbb{Z}}$,
which is stronger than requiring entropy rate $h>0$.
\begin{definition}
  A~process $(X_i)_{i\in\mathbb{Z}}$ is called a~finite-energy process
  if
  \begin{align*}
    P(X_{t+\abs{w}+1}^{t+\abs{wu}}=u|X_{t+1}^{t+\abs{w}}=w)\le Kc^{\abs{u}}
  \end{align*}
  for all $t\in\mathbb{Z}$, all $u,w\in\mathbb{X}^*$, and certain
  constants $c<1$ and $K$, as long as $P(X_{t+1}^{t+\abs{w}}=w)>0$.
\end{definition}
The term ``finite-energy process'' has been coined by
Shields.\cite{Shields97} We are unaware of the motivation for this
name.

Now let us discuss the adopted model of words. It uses admissible
grammars mentioned in Subsection \ref{ssecAdmissible}. A~function
$\Gamma$ such that $\Gamma(w)$ is a~grammar and generates language
$\klam{w}$ for each string $w\in\mathbb{X}^+$ is called a~grammar
transform.\cite{KiefferYang00} Any such grammar $\Gamma(w)$ is
admissible and is given by its set of production rules
\begin{align}
\label{FullGrammar}
\Gamma(w)&=   \klam{ 
\begin{array}{l}
A_1\rightarrow\alpha_1, \\
A_2\rightarrow\alpha_2, \\
..., \\
A_n\rightarrow\alpha_n 
    \end{array}
}, 
\end{align}
where $A_1$ is the start symbol, other $A_i$ are secondary
nonterminals, and the right-hand sides of rules satisfy $\alpha_i\in
(\klam{A_{i+1},A_{i+2},...,A_n}\cup\mathbb{X})^*$.  
The number of distinct nonterminal symbols in grammar
(\ref{FullGrammar}) will be called the vocabulary size of $\Gamma(w)$ and
denoted by
\begin{align}
  \voc{\Gamma(w)}:=\card \klam{A_{1},A_{2},...,A_n} =n
  .
\end{align}

In the following, let us consider vocabulary size of admissibly
minimal grammar transforms, which were defined exactly in the previous
paper.\cite{Debowski11b} The formal definition is too long to quote
here but, briefly speaking, admissibly minimal grammar transforms
minimize a~certain nice length function of grammars. A~simple example
of a~grammar length function is Yang-Kieffer length
\begin{align}
\label{YKlength}
\abs{\Gamma(w)}:=\sum_{i=1}^n \abs{\alpha_i}
\end{align}
for grammar (\ref{FullGrammar}), where $\abs{\alpha_i}$ is the length
of the right-hand side of rule
$A_i\rightarrow\alpha_i$.\cite{CharikarOthers05} 

In our application we use a~slightly different length function
$\abs{\abs{\Gamma(w)}}$, which measures the length of $\Gamma(w)$
after a~certain reversible binary encoding, and we choose a~grammar
transform that minimizes $\abs{\abs{\Gamma(w)}}$ for a~given string
$w$.  Nonterminals of these so called admissibly minimal grammar
transforms often correspond to words in the linguistic
sense.\cite{DeMarcken96,KitWilks99} Thus we stipulate that the
vocabulary size of an admissibly minimal grammar is close to the
number of distinct words in the text.

The formalization of Proposition (H) is as follows:
\begin{theorem}
  \label{theoQLThesis}
  Let $(X_i)_{i\in\mathbb{Z}}$ be a stationary finite-energy strongly
  nonergodic process over a~finite alphabet $\mathbb{X}$. If
 \begin{align}
    \label{QLPremise}
    \liminf_{n\rightarrow\infty} \frac{\card{U_\delta(n)}}{n^\beta}>0
  \end{align}
  holds for some $\beta\in (0,1)$ and $\delta\in(\frac{1}{2},1)$ then
  \begin{align}
    \label{QLClaim}
    \limsup_{n\rightarrow\infty} 
    \sred_P{\okra{\frac{
          \voc{\Gamma(X_{1}^{n})}
        }{
          n^\beta(\log n)^{-1}
        }}^p}>0
    ,
    \quad
    p>1
    ,
  \end{align}
  for any admissibly minimal grammar transform $\Gamma$.
\end{theorem}

There are also two similar theorems that link inequalities
(\ref{QLPremise}) and (\ref{QLClaim}) with Hilberg's conjecture. These
are formalizations of Propositions (H') and (H'') respectively.
\begin{theorem}
  \label{theoQLThesisA}
  Let $(X_i)_{i\in\mathbb{Z}}$ be a stationary strongly nonergodic
  process over a~finite alphabet $\mathbb{X}$.  If (\ref{QLPremise})
  holds for some $\beta\in (0,1)$ and $\delta\in(\frac{1}{2},1)$ then
  we have
\begin{align}
  \label{QLClaimA}
  \limsup_{n\rightarrow\infty} \frac{E_\mu(n)}{n^\beta}>0
  .
\end{align}
\end{theorem}
\begin{theorem}
  \label{theoQLThesisB}
  Let $(X_i)_{i\in\mathbb{Z}}$ be a stationary finite-energy process
  over a~finite alphabet $\mathbb{X}$. Assume that
  \begin{align}
    \label{QLPremiseB}
    \liminf_{n\rightarrow\infty} \frac{E_\mu(n)}{n^\beta}>0
  \end{align}
  holds for some $\beta\in(0,1)$. Then we have (\ref{QLClaim})
  for any admissibly minimal grammar transform $\Gamma$.
\end{theorem}

Theorem \ref{theoQLThesis} does not follow from Theorems
\ref{theoQLThesisA} and \ref{theoQLThesisB} because (\ref{QLClaimA})
is a~weaker condition than (\ref{QLPremiseB}). However, all these
propositions are true and the proofs of these propositions are almost
simultaneous.\cite{Debowski11b} By an easy argument, using Lemma
\ref{theoExcessBound} from Subsection \ref{ssecMeasureMI}, it can be
also shown that $n$-symbol excess entropy $E(n)$ in Theorems
\ref{theoQLThesis}--\ref{theoQLThesisB} may be replaced with expected
algorithmic information $\sred_P{I(X_{1}^{n};X_{n+1}^{2n})}$.

\subsection{The zoo of Santa Fe processes}

Now I~will present a~few stochastic processes to which my theorems may
be applied.\cite{Debowski09,Debowski10,Debowski11b,Debowski11c} These
processes are merely simple mathematical models that satisfy
hypotheses of Theorems \ref{theoQLThesis}, \ref{theoQLThesisA}, and
\ref{theoQLThesisB}. They model some aspects of human communication
but they do not pretend to be very realistic models of language. The
purpose of these constructions is to enhance our imagination and to
show that the hypotheses of the theorems can be satisfied.

Quite early in my investigations I came across the following
process. Let the alphabet be $\mathbb{X}=\mathbb{N}\times\klam{0,1}$
and let the process $(X_i)_{i\in\mathbb{Z}}$ have the form
\begin{align}
  \label{exUDP}
  X_i:=(K_i,Z_{K_i})
  ,
\end{align}
where $(Z_k)_{k\in\mathbb{N}}$ and $(K_i)_{i\in\mathbb{Z}}$ are
probabilistically independent. Moreover, let $Z_k$ be IID with
marginal distribution (\ref{Zk}) and let $(K_i)_{i\in\mathbb{Z}}$ be
such an ergodic stationary process that $P(K_i=k)>0$ for every natural
number $k\in\mathbb{N}$.  Under these assumptions it can be
demonstrated that $(X_i)_{i\in\mathbb{Z}}$ forms a~strongly nonergodic
process.\cite{Debowski09} I~will call process $(X_i)_{i\in\mathbb{Z}}$
with variables $X_i$ as in (\ref{exUDP}) the Santa Fe process because
I~discovered it during my visit to the Santa Fe Institute.

Santa Fe process (\ref{exUDP}) can be interpreted as a~sequence of
statements which describe a~fixed random object
$(Z_k)_{k\in\mathbb{N}}$ in a~repetitive and consistent way. Each
statement $X_i=(k,z)$ reveals both the address $k$ of a~random bit of
$(Z_k)_{k\in\mathbb{N}}$ and its value $Z_k=z$.  The description is
consistent, namely, if two statements $X_i=(k,z)$ and $X_j=(k',z')$
describe the same bits ($k=k'$) then they always assert identical
value ($z=z'$). 

Moreover, we can see that the revelation of the bit address is
important to assure the existence of functions $s_k$ such that
(\ref{UDPcondi}) holds. Indeed we may take
\begin{align}
  \label{funcUDP}
  s_{k}(v):=
  \begin{cases}
    0 & \text{if $(k,0)\subseq v$ and $(k,1)\not\subseq v$}, \\
    1 & \text{if $(k,1)\subseq v$ and $(k,0)\not\subseq v$}, \\
    2 & \text{else},
  \end{cases}
\end{align}
where we write $u\subseq v$ when a~sequence $v$ contains
string $u$ as a~substring.

For these functions $s_k$, I~have shown\cite{Debowski11b} that the
cardinality of set $U_\delta(n)$ obeys
\begin{align}
  \label{UDPPowerLawX}
  \card{U_\delta(n)}\ge
  \kwad{\frac{n}{-\zeta(\beta^{-1})\log(1-\delta)}}^{\beta}
\end{align}
for process (\ref{exUDP}) if variables $K_i$ are IID and power-law
distributed,
\begin{align}
  \label{ZetaK}
  P(K_i=k)&=k^{-1/\beta}/\zeta(\beta^{-1})
  ,
  &
  \beta&\in(0,1)
  ,
\end{align}
where $\zeta(\alpha)=\sum_{k=1}^\infty k^{-\alpha}$ is the zeta
function. 

In contrast, it can be seen that the cardinality of set $U_\delta(n)$
is of order $\log n$ if $(X_i)_{i\in\mathbb{Z}}$ is a~Bernoulli
process with binary variables $X_i:\Omega\rightarrow\klam{0,1}$,
a~random parameter $\Theta=\sum_{k=1}^\infty 2^{-k} Z_k$, and
conditional distribution
\begin{align}
 P(X_1^n||\Theta)=\prod_{i=1}^n
 \Theta^{X_i}(1-\Theta)^{1-X_i}
 .
\end{align}

Next, let us discuss a~certain modification of the Santa Fe process.
As I~have said before, facts that are mentioned in texts repeatedly
fall roughly under two types: (a) facts about objects that do not
change in time (like mathematical or physical constants), and (b)
facts about objects that evolve with a~varied speed (like culture,
language, or geography). An attempt to model the latter phenomenon
leads to processes that are mixing, as we will see now.

In the following, let us replace individual variables $Z_k$ in the
Santa Fe process with Markov chains $(Z_{ik})_{i\in\mathbb{Z}}$. The
Markov chains are formed by iterating a~binary symmetric channel.
Consecutively, let us put
\begin{align}
  \label{exMixing}
  X_i&=(K_i,Z_{i,K_i})
  ,
\end{align}
where processes $(K_i)_{i\in\mathbb{Z}}$ and
$(Z_{ik})_{i\in\mathbb{Z}}$, where $k\in\mathbb{N}$, are independent
and distributed as follows.  First, variables $K_i$ are distributed
according to formula (\ref{ZetaK}), as before. Second, each process
$(Z_{ik})_{i\in\mathbb{Z}}$ is a~Markov chain with marginal
distribution 
  \begin{align}
    \label{Zik}
    P(Z_{ik}=0)=P(Z_{ik}=1)=\frac{1}{2}
  \end{align}
and cross-over probabilities
\begin{align}
  P(Z_{ik}=z|Z_{i-1,k}=1-z)&=p_k
  ,
  &
  z\in\klam{0,1}
  .
\end{align}

The random object $(Z_{k})_{k\in\mathbb{N}}$ described by original
Santa Fe process (\ref{exUDP}) does not evolve, or rather, no bit
$Z_{k}$ is ever forgotten once revealed. In contrast, the random
object $(Z_{ik})_{k\in\mathbb{N}}$ described by modified Santa Fe
process (\ref{exMixing}) is a~function of time $i$ and the probability
that the $k$-th bit flips at a~given instant equals $p_k$. For
$p_k=0$, process (\ref{exMixing}) collapses to process (\ref{exUDP}).

As I~have shown previously,\cite{Debowski11c} the modified Santa Fe
process defined in (\ref{exMixing}) is mixing for $p_k\in(0,1)$, and
thus ergodic. Moreover, for $p_k\in[0,1]$, I~have also demonstrated
asymptotics
\begin{align}
    \label{limEnUDP}
    \lim_{n\rightarrow\infty} \frac{E_\mu(n)}{n^\beta}=
    \frac{(2-2^\beta)\Gamma(1-\beta)}{[\zeta(\beta^{-1})]^\beta}
\end{align}
if $\lim_{k\rightarrow\infty} p_k/P(K_i=k)=0$ and $K_i$ obey law
(\ref{ZetaK}).\cite{Debowski11c} In the equation above
$\Gamma(z)=\int_0^\infty t^{z-1}e^{-t}dt$ is the gamma
function. Formula (\ref{limEnUDP}) follows from approximating an exact
expression for $E_\mu(n)$ with an integral. Note that (\ref{limEnUDP})
holds also in the case of original strongly nonergodic Santa Fe
process (\ref{exUDP}).

Neither of processes defined so far is a~process over a~finite
alphabet, as required in Theorems
\ref{theoQLThesis}--\ref{theoQLThesisB}. To construct the desired
processes over a~ternary alphabet, I~have used stationary (variable
length) coding of processes over one alphabet into processes over
another alphabet. This transformation preserves stationarity,
(non)ergodicity, and entropy---to some
extent.\cite{Debowski10,Debowski11c} Despite elaborate notation, the
idea of this transformation is quite simple.

First, let a~function $f:\mathbb{X}\rightarrow\mathbb{Y}^*$, called
a~coding function, map symbols from alphabet $\mathbb{X}$ into strings
over another alphabet $\mathbb{Y}$. We define its extension to double
infinite sequences $f^{\mathbb{Z}}:\mathbb{X}^{\mathbb{Z}}\rightarrow
\mathbb{Y}^{\mathbb{Z}}\cup(\mathbb{Y}^*\times\mathbb{Y}^*)$ as
\begin{align}
  \label{InfExtension}
  f^{\mathbb{Z}}((x_i)_{i\in\mathbb{Z}})&:=
  ... f(x_{-1})f(x_{0})\textbf{.}f(x_1)f(x_2)...
  ,
\end{align}
where $x_i\in\mathbb{X}$ and the bold-face dot separates the $0$-th
and the first symbol. Then for a~stationary process
$(X_i)_{i\in\mathbb{Z}}$ with variables
$X_i:\Omega\rightarrow\mathbb{X}$, we define process
\begin{align}
  \label{DefY}
  (Y_i)_{i\in\mathbb{Z}}:=f^{\mathbb{Z}}((X_i)_{i\in\mathbb{Z}})
\end{align}
with variables $Y_i:\Omega\rightarrow\mathbb{Y}$.

In the following application, let us assume the infinite alphabet
$\mathbb{X}=\mathbb{N}\times\klam{0,1}$, the ternary alphabet
$\mathbb{Y}=\klam{0,1,2}$, and the coding function
\begin{align}
  \label{ConjCode}
  f(k,z)=b(k)z2
  ,
\end{align}
where $b(k)\in\klam{0,1}^+$ is the binary representation of a~natural
number $k$ stripped of the leading digit $1$.

Transformation (\ref{DefY}) does not preserve stationarity in general
but process $(Y_i)_{i\in\mathbb{Z}}$ is asymptotically mean stationary
(AMS) for process (\ref{exMixing}) and coding function
(\ref{ConjCode}).\cite{Debowski10} Then for the distribution
$\nu=P((Y_i)_{i\in\mathbb{Z}}\in\cdot\,)$ and the shift operation
$T((y_i)_{i\in\mathbb{Z}}):=(y_{i+1})_{i\in\mathbb{Z}}$ there exists
a~stationary measure
\begin{align}  
  \label{BarMu}
  \bar\nu(A):=\lim_{n\rightarrow\infty} \frac{1}{n}  \sum_{i=0}^{n-1}
  \nu\circ T^{-i}(A)
  ,
\end{align}
called the stationary mean of $\nu$ \cite{GrayKieffer80,Debowski10}.
It is convenient to suppose that probability space
$(\Omega,\mathcal{J},P)$ is rich enough to support a~process $(\bar
Y_i)_{i\in\mathbb{Z}}$ with the distribution $\bar\nu=P((\bar
Y_i)_{i\in\mathbb{Z}}\in\cdot\,)$. Process $(\bar
Y_i)_{i\in\mathbb{Z}}$ will be called the stationary coding of
$(X_i)_{i\in\mathbb{Z}}$.

Processes $(X_i)_{i\in\mathbb{Z}}$, $(Y_i)_{i\in\mathbb{Z}}$, and
$(\bar Y_i)_{i\in\mathbb{Z}}$ have isomorphic shift-invariant algebras
for some nice coding functions, called synchronizable injections.
\cite{Debowski10} Coding function (\ref{ConjCode}) is an instance of
such an injection. Thus processes $(Y_i)_{i\in\mathbb{Z}}$ and $(\bar
Y_i)_{i\in\mathbb{Z}}$ obtained from process (\ref{exMixing}) using
(\ref{ConjCode}) are nonergodic if $p_k=0$ and ergodic if
$p_k\in(0,1)$.

Now let us consider block mutual information for the stationary coding
$(\bar Y_i)_{i\in\mathbb{Z}}$ of process (\ref{exMixing}) using coding
function (\ref{ConjCode}).  I~have shown\cite{Debowski11c} that
\begin{align}
  \label{liminfEBarn}
  \liminf_{m\rightarrow\infty} \frac{E_{\bar\nu}(m)}{(m\log^{-1} m)^\beta}>0
  ,
\end{align}
for $E_{\bar\nu}(m)=I_P\okra{\bar Y_{1:m};\bar Y_{m+1:2m}}$ and
cross-over probabilities $p_k\le P(K_i=k)$. This bound should be
contrasted with inequality 
\begin{align}
  \label{limsupEBarn}
  \limsup_{m\rightarrow\infty} \frac{E_{\bar\nu}(m)}{m^\beta}>0
\end{align}
which follows for $p_k=0$.\cite{Debowski10,Debowski11b} It is an
interesting open problem whether (\ref{limsupEBarn}) can be
generalized for $p_k>0$.

Another interesting open problem concerns the question whether the
stationary coding is a~finite-energy process. This property is assumed
in Theorems \ref{theoQLThesis} and \ref{theoQLThesisB} to bound the
length of the longest repeat and hence to bound the vocabulary
size.\cite{Debowski11b} I~have shown that process $(\bar
Y_i)_{i\in\mathbb{Z}}$ is a~finite-energy process for
$\beta>0.7728...$ and $p_k=0$\cite{Debowski10} but I~wonder if it also
holds for other exponents $\beta$ and cross-over probabilities $p_k$.

\section{Afterthoughts for theoreticians}
\label{secTheory}

The translation of abstract mathematical results back into linguistic
reality can be challenging.  In the following, I~want to share a~few
remarks about theoretical limitations of my constructions as models of
natural language. This part of the paper is born by typical comments
I~receive about my model.

\subsection{What are those `facts'?}

Many people to whom I~have presented the concept of Santa Fe processes
ask the question: ``What are those `facts'?'' Whereas in models
(\ref{exUDP}) and (\ref{exMixing}) the facts are just some binary
variables, I~tried to interpret these variables in Section
\ref{ssecDescriptions} as independent propositions about particular
complicated infinite random objects, consistently described in the
texts. These objects might be static like a~mathematical or physical
constant or might evolve slowly like cultural heritage. However, the
identification of the sequence of independent facts described in the
actual texts in natural language is left as a~matter of future
research.

This does not mean that nothing can be said about the interpretation
of facts at the moment.  Let me make an important remark.  In my
model, the probability of mentioning independent propositions in texts
obeys a~power law. If the same applies to natural language, it seems
unlikely that the mentioned facts are the binary digits of halting
probability $\Omega$, which has an appealing property of representing
a~large body of mathematical knowledge in a~concise
form.\cite{Chaitin75,Gardner79} Although the digits of $\Omega$ have
been proved to be in a~sense independent (i.e., algorithmically
random), I~suppose that information relayed by humans in a~repetitive
way is mostly unrelated to $\Omega$ because human beings do not have
supernatural powers to guess the bits of $\Omega$ at a~power-law rate.
The facts that are usually mentioned in texts should concern `more
everyday' objects.

\subsection{Are facts and words the same?}

Another type of reaction I~have heard is: ``But facts and words are
the same so your result about the implication of power laws for them
is a~tautology!'' My short answer to the criticism is this: ``Words
and facts are very different entities, however, so my result is
nontrivial.''

To support this reply let us notice the following. F. de Saussure made
a~famous observation that a~linguistic sign is a~pair of a~word (i.e.,
a string of characters) and a~meaning (roughly, an object to which the
word refers).\cite{Saussure16} To a large extent, the mapping between
words and objects is one-to-one. Therefore,
\begin{align*}
  \dzi{\text{number of referred objects}} \approx \dzi{\text{number of
      words}} .
\end{align*}
In contrast, I~have claimed a~relationship
\begin{align*}
  \dzi{\text{number of words}} \gtrsim 
  \frac{
    \dzi{\text{number of independent facts}}
  }{
    \log\dzi{\text{length of text}}
  }
  .
\end{align*}
That inequality can be strictly sharp because objects (say, things,
concepts, qualities, or activities) are different entities than facts
(i.e., propositions which assume binary values). The inequality is
also nontrivial because propositions usually consist of more than one
word.

\subsection{Finite active vocabulary and division of knowledge }

An important limitation of my results is their asymptotic
character. I~have dealt with asymptotic statements because it is
simpler to work out a~mathematical model in that case. In reality,
however, the number of different words actively used by a~single
person is of order $r(w)\approx 10^3\div 10^4$. For word ranks below
that value, Zipf's law (\ref{ZipfMandelbrot}) is observed with
$\alpha\approx 1$. In contrast, for larger word ranks, word
frequencies decay exponentially in collections of texts written by
a~single author.\cite{MontemurroZanette02} It is not known whether
a~similar breakdown arises for vocabulary of admissibly minimal
grammars or for Hilberg's law (\ref{HilbergGenII}).  This question is
worth investigating.

Whereas word frequencies decay ultimately exponentially for
single-author collections of texts, a~different relationship is
observed in multi-author text collections. Namely, for $r(w)\gtrsim
10^3\div 10^4$, the exponent in the Zipf-Mandelbrot law
(\ref{ZipfMandelbrot}) switches to $\beta\approx 0.4$ rather than to
$\beta\approx 0$.\cite{FerrerSole01b,MontemurroZanette02} This
phenomenon can be interpreted as developing social structures to
maintain and transmit a~larger body of knowledge than any individual
could manage on his or her own.
To model this phenomenon properly we should assume that
finite texts produced by single authors are woven up into a~discourse
(a communication network) of yet unrecognized topology, rather than
concatenated in an arbitrary infinite sequence
$(X_i)_{i\in\mathbb{Z}}$.


\subsection{How does language differ to  maths, music, and DNA?}

Texts in natural language are not the only type of a~complex
communication system that occurs in nature. Examples of other systems
are musical transcripts, mathematical writings, computer programs, or
genome (DNA and RNA).  One may investigate quantitative laws obeyed in
these systems, just as it is done for natural
language.\cite{EllisHitchcock86,WangOthers08} Moreover, although the
notion of a~word is connected to linguistics, one may investigate
Hilberg's conjecture and statistical properties of admissibly minimal
grammars for any symbolic sequence. One may also try to interpret or
predict respective experimental results theoretically.

For example, Ebeling et al.\ estimated $n$-symbol entropy by counting
$n$-tuples in samples of texts in natural language and classical
music.  They confirmed formula (\ref{HilbergGen}) for $n\le 15$
characters with $\beta\approx 0.5$ for natural language texts and
$\beta\approx 0.25$ for classical music
transcripts.\cite{EbelingNicolis92,EbelingPoschel94}

Mathematical writings are another interesting communication system
which has not been much researched from a~quantitative perspective.
I~suppose that mathematical writings obey relaxed Hilberg formula
(\ref{HilbergGenII}) similar to that of music or novels in natural
language because all these symbolic sequences are produced by humans
for humans, either for their work or entertainment. In either case,
I~suppose that humans need a~large degree of repetition to learn from
an information source how to react to it properly. Hence relationship
(\ref{HilbergGenII}) should arise. Intuitively, our abilities to use
a~particular language, to enjoy a~particular style of music, or to
work in a~branch of mathematics are all learned and learning is only
possible if there are some patterns to be learned.

In contrast, computer programs or DNA are sequences that control
behavior of machines like computers or biological cells.  These
machines can interpret control sequences in a~fixed manner without
learning or loss of synchronization caused by other factors. Hence
there is less need of repetition in the control
sequences. Consequently, relationship (\ref{HilbergGenII}) and
word-like structures need not arise in compiled computer programs or
DNA to such an extent as in typical texts in natural language.

\section{Afterthoughts for experimentalists}
\label{secExperiments}

The body of theoretical insight gathered so far asks for experimental
verification. One can rightly question to what extent these results
may be applied to real texts in natural language.  The preparation of
a~sound experimental study requires much more space than it is left in
this paper, so let me only sketch a~few problems. There are several
levels of experimental verification, which correspond to growing
difficulty.

The easiest thing to do is to check whether Zipf's or Herdan's law is
satisfied for languages in which words are delimited by spaces. There
are plenty of articles about that, including observations that Zipf's
law breaks for large ranks.\cite{FerrerSole01b,MontemurroZanette02} In
contrast, it is a~bit harder to verify whether a~power-law is
satisfied for nonterminals in admissibly minimal grammars. The next
task is to verify Hilberg's conjecture for a~particular text. In the
end, the hardest thing to do is to estimate how many facts are jointly
described in two given texts.

In this section, I~will touch some of these questions. First,
I~discuss some grammar transforms that \emph{may not} be used to
approximate admissibly minimal grammars or to detect word boundaries.
Second, I~suggest how mutual information can be efficiently
estimated.

\subsection{What is the appropriate grammar-based code?}

I~have claimed that there is a~tight relationship between the number
of distinct words in a~text in natural language and the number of
distinct nonterminal symbols in an admissibly minimal grammar for the
text. Although this claim is supported by several computational
experiments,\cite{Wolff80,DeMarcken96,KitWilks99} the regression
between these two quantities has not been surveyed directly so far. In
fact, investigating this dependence is hard because computing
admissibly minimal grammars is extremely costly, even in
approximation.\cite{DeMarcken96,KitWilks99} In contrast,
computationally less intensive grammar transforms may detect spurious
structures.

For example, irreducible grammar
transforms\cite{NevillManning96,KiefferYang00,CharikarOthers05}
exhibit a~power-law growth of vocabulary size for any source with
a~positive entropy rate.\cite{Debowski07d} To see it, let us first
observe inequality
\begin{align}
\label{IrredIneq}
  \abs{\mathsf{G}}-\voc{\mathsf{G}} \le (\voc{\mathsf{G}}+\card\mathbb{X})^2,
\end{align}
where $\mathsf{G}$ is an irreducible grammar\cite{KiefferYang00} and
$\abs{\mathsf{G}}$ is the Yang-Kieffer length of
$\mathsf{G}$.\cite{Debowski11b} Any irreducible grammar satisfies
(\ref{IrredIneq}) since any concatenation of two symbols may only
occur on the right-hand sides of its rules only once.

What happens if an irreducible grammar $\mathsf{G}$ compresses a~text
of length $n$ produced by a~stationary source with entropy rate
$h_\mu$? Then we obtain $\abs{\mathsf{G}}\gtrsim h_\mu n/\log n$ from
the source coding inequality $\abs{\mathsf{G}}\log
\abs{\mathsf{G}}\gtrsim h_\mu n$ and the trivial inequality
$\abs{\mathsf{G}}\le n$. Combining that with (\ref{IrredIneq}) yields
the power law
\begin{align}
\label{IrredIneqII}
\voc{\mathsf{G}}\gtrsim \sqrt{\frac{h_\mu n}{\log n}} -\card\mathbb{X}
- 1.
\end{align}
In particular, the higher the entropy rate is, the more nonterminals
are detected by the grammar.

In my opinion relationship (\ref{IrredIneqII}) is an artifact. It
arises because irreducible grammars minimize a~wrongly chosen length
function. If we choose a~certain different grammar length
function\cite{DeMarcken96,Debowski11b} then, after complete
minimization, we obtain admissibly minimal grammars. The number of
nonterminals in approximations of such grammars is a~thousand times
larger for texts in natural language than for IID
sources.\cite{Debowski07d} Thus I~suppose that the vocabulary size of
admissibly minimal grammars is lower-bounded by the $n$-symbol excess
entropy rather than by the entropy rate.


\subsection{How to measure mutual information?}
\label{ssecMeasureMI}

Entropy of a~long sequence of random variables is hard to estimate. It
can be effectively bounded only from above and there is some
systematic nonnegligible error term. We know that an upper bound for
the entropy of a~text is given by the expected length of any
prefix-free code for the text. However, the length of the shortest
effectively decodable code for the text equals the algorithmic
complexity of the text, which is greater than the entropy. Thus any
intention of estimating Shannon entropy by universal coding ends up
with estimating algorithmic complexity.

Now, certain care must be given to distinguishing Shannon entropy
$H_P(X_1^n)$ and algorithmic complexity $H(X_1^n)$.  Although we have
inequality (\ref{BayesianEBounds}) for any computable measure $P$, the
difference
\begin{align}
  \label{HHP}
  \sred_P H(X_1^n)-H_P(X_1^n)
\end{align}
can exceed any sublinear function of $n$ if $P$ is stationary but not
computable.\cite{Shields93}

Whereas the classical proof of unboundedness of (\ref{HHP}) is
difficult,\cite{Shields93} a~similar result can be obtained using
Santa Fe process (\ref{exUDP}).  Let $P$ be the probability measure
for process (\ref{exUDP}) and let
\begin{align}
 F=P(\,\cdot\,|(Z_k)_{k\in\mathbb{N}}) 
\end{align}
be its conditional measure.\cite{Billingsley79,Debowski09} The values
of conditional measure $F$ depend on the value of process
$(Z_k)_{k\in\mathbb{N}}$. In particular, measure $F$ is not computable
for algorithmically random $(Z_k)_{k\in\mathbb{N}}$ (i.e., $F(X_1^n)$
cannot be computed given no $(Z_k)_{k\in\mathbb{N}}$). Further,
\begin{align}
  H_P(X_1^n)-\sred_P H_F(X_1^n)
  &=I_P(X_1^n;(Z_k)_{k\in\mathbb{N}})
  \nonumber
  \\
  &=O(n^\beta)
  .
\end{align}
Moreover, we have the source coding inequality $H_P(X_1^n) \le \sred_P
H(X_1^n) = \sred_P\sred_F H(X_1^n)$.  Hence we obtain
\begin{align}
  \sred_P\kwad{\sred_F H(X_1^n)-H_F(X_1^n)}\ge O(n^\beta)
\end{align}
as the desired result, where an analogue of (\ref{HHP}) appears.
In other words, universal coding bounds suffer from a~large systematic
error for Shannon entropy of noncomputable probability measures.

Now I~will show that the error of the coding bounds can be greatly
reduced for certain noncomputable measures when we rather bound
algorithmic complexity. This opens way to bounding also algorithmic
information, which is a~difference of complexities.

Suppose that
\begin{align}
  \label{Bayesian}
  P=\int P(\,\cdot\,|\Theta) dP
\end{align}
where $P(\,\cdot\,|\Theta)$ are measures of stationary Markov chains
for particular values of transition probabilities $\Theta$ and
$P(\Theta\in\cdot\,)$ is an appropriate prior over all transition
probabilities for all possible orders of Markov chains. Measures
$P(\,\cdot\,|\Theta)$ are not computable for algorithmically random
$\Theta$.  It is likely, however, that the Shannon-Fano code yielded
by measure (\ref{Bayesian}) is computable and universal, i.e., $-\log
P(X_1^n)$ can be computed given $X_1^n$ and
\begin{align}
  \label{Universal}
  \lim_{n\rightarrow\infty}
  \frac{1}{n}\sred_Q\kwad{-\log P(X_1^n)}=h_\nu
\end{align}
holds for any stationary measure
$\nu=Q((X_i)_{i\in\mathbb{Z}}\in\cdot\,)$.

Consider now pointwise mutual information
\begin{align}
  I^P(x_1^n;&\,x_{n+1}^{2n}):= 
  \nonumber
  \\
  \label{IaboveP}
  &H^P(x_1^n)+H^P(x_{n+1}^{2n})-H^P(x_1^{2n})
  ,
\end{align}
using pointwise entropy
\begin{align}
  H^P(x_1^n):= -\log P(X_1^n=x_1^n)
  .
\end{align}
The Shannon-Fano coding yields inequality
\begin{align}
  \label{HHaboveP}
  H(x_1^n)\le H^{P}(x_1^n) + C^{P}_n
\end{align}
where $x_1^n$ is arbitrary and $C^{P}_n=c_{P} + 2\log n$ for a~certain
constant $c_P$.  Thus we define the loss of pointwise mutual
information with respect to algorithmic information as
\begin{align}
  L^P(x_1^n;&\,x_{n+1}^{2n}):= 
  \label{LaboveP}
  I^P(x_1^n;x_{n+1}^{2n})-I(x_1^n;x_{n+1}^{2n}).
\end{align}

We will next use the following lemma:
\begin{lemma}
\label{theoExcessBound}
Consider a~function $G:\mathbb{N}\rightarrow\mathbb{R}$ such that
$\lim_k G(k)/k=0$ and $G(n)\ge 0$ for all but finitely many $n$. Then
for infinitely many $n$, we have $2G(n)- G(2n)\ge
0$.\cite{Debowski11b}
\end{lemma}
If (\ref{Universal}) holds indeed then equality
\begin{align}
  \lim_{n\rightarrow\infty}
  \frac{1}{n}\sred_Q H(X_1^n)=h_\nu
\end{align}
for any stationary measure $Q$ and Lemma \ref{theoExcessBound} for
function $G(n)= \sred_Q \kwad{H^{P}(X_1^n) - H(X_1^n)} + C^{P}_n$ yield
\begin{align}
  \label{limsupLI}
  \limsup_{n\rightarrow\infty} 
  \sred_Q \kwad{L^P(X_1^n;X_{n+1}^{2n}) + C^{P}_n}\ge 2\log 2
  .
\end{align}
Hence, as long as (\ref{Universal}) holds, pointwise mutual
information (\ref{IaboveP}) is an upper estimate of algorithmic
information, up to a~small logarithmic correction $C^{P}_n$.

In certain cases of noncomputable $Q$, quantity (\ref{IaboveP}) is
also a~lower estimate of algorithmic information. Observe that $P$ is
computable.  Hence for all $P$-algorithmically random sequences
$(x_i)_{i\in\mathbb{N}}$ we have by definition
\begin{align}
  \label{MartinLof}
  \inf_{n\in\mathbb{N}}  
  \kwad{H(x_1^n)+\log P(x_{1}^n)}> -\infty
  .
\end{align}
Moreover, we have the following fact:
\begin{theorem}
  The set of $P$-algorithmically random sequences is the union of sets
  of $P(\cdot|\Theta)$-algorithmically random sequences over all
  parameters $\Theta$ that are algorithmically random against the
  prior
  $P(\Theta\in\cdot\,)$.\cite{VovkVyugin93,VovkVyugin94,Takahashi08}
\end{theorem}

Each of those sets of algorithmically random sequences has the
respective full measure, so in a~sense it contains all outcomes
typical of that measure.  Let us fix a~sequence
$(x_i)_{i\in\mathbb{N}}$ belonging to one of these sets.  Because $P$
is stationary, by (\ref{HHaboveP}) and (\ref{MartinLof}), we obtain
\begin{align}
  \label{limsupLII}
  \sup_{n\in\mathbb{N}} \kwad{L^{P}(x_1^n;x_{n+1}^{2n})
    - 2C^{P}_n}<\infty
  .
\end{align}
Hence, inequality (\ref{limsupLII}) gives an upper bound for loss
(\ref{LaboveP}) for typical realizations of typical Markov chains.

In view of bounds (\ref{limsupLI}) and (\ref{limsupLII}), pointwise
mutual information (\ref{IaboveP}) could be considered an interesting
estimate of algorithmic information. It could be used for verifying
(or rather falsifying) relaxed algorithmic Hilberg conjecture
(\ref{HilbergGenII}). The details of computing distribution $P$ and
pointwise mutual information (\ref{IaboveP}) will be worked out in
another paper, however. Although the sketched distribution $P$ is
computable, there are some problems with assuring \emph{efficient}
computability.

\section{Conclusion}
\label{secConclusion}

In this article I~have presented a~wide array of interesting issues
that arise when quantitative research in language is combined with
fundamental research in information theory. As I~have indicated in the
Introduction, the presented ideas may be inspiring for more general
studies of complex systems because of the demonstrated connections
among excess entropy, power laws, and the emergence of hierarchical
structures in data.

\begin{acknowledgments}
  I~wish to thank Jacek Koronacki, Jan Mielniczuk, Jim Crutchfield,
  and two anonymous referees for valuable comments.
\end{acknowledgments}


%

\end{document}